\documentclass[conference]{IEEEtran}
\IEEEoverridecommandlockouts
\usepackage{cite}
\usepackage{amsmath,amssymb,amsfonts}

\usepackage{graphicx}
\usepackage{textcomp}
\usepackage{xcolor}
\usepackage[ruled,vlined ]{algorithm2e}
\usepackage{amsmath,amsfonts,amsthm,amssymb,bm, verbatim,dsfont,mathtools}
\usepackage[mathcal]{euscript}
\usepackage{color,graphicx}
\usepackage{subfigure}
\usepackage{etoolbox}
\usepackage{tikz}
\usepackage{xr,xspace}
\usepackage{todonotes}
\usepackage{paralist}
\usepackage{caption,soul}

\usepackage{hyperref}
\usepackage{accents}

\newtheorem{theorem}{Theorem}

\newcommand{\expect}[1]{\mathbb{E}\left[ #1 \right]}
\newlength{\dhatheight}
\newcommand{\doublehat}[1]{%
    \settoheight{\dhatheight}{\ensuremath{\hat{#1}}}%
    \addtolength{\dhatheight}{-0.35ex}%
    \hat{\vphantom{\rule{1pt}{\dhatheight}}%
    \smash{\hat{#1}}}}

\def\BibTeX{{\rm B\kern-.05em{\sc i\kern-.025em b}\kern-.08em
    T\kern-.1667em\lower.7ex\hbox{E}\kern-.125emX}}
\begin{document}

\title{Multi-User MABs with User Dependent Rewards for Uncoordinated Spectrum Access\\ 
}

\author{\IEEEauthorblockN{Akshayaa Magesh\IEEEauthorrefmark{1},
Venugopal V. Veeravalli \IEEEauthorrefmark{2} \thanks{This research was supported by the US NSF SpecEES under grant number 1730882, through the University of Illinois at Urbana-Champaign.}}
\IEEEauthorblockA{Department of Electrical and Computer Engineering \\
Coordinated Science Laboratory \\
University of Illinois at Urbana-Champaign\\
Emails: \IEEEauthorrefmark{1}amagesh2@illinois.edu,
\IEEEauthorrefmark{2}vvv@illinois.edu}}
\maketitle

\begin{abstract}
The uncoordinated spectrum access problem is studied using a multi-player multi-armed bandits framework. In the considered system model there is no central control and the users cannot communicate with each other. Furthermore,  the environment may appear differently to different users, \textit{i.e.}, the mean rewards as seen by different users for a particular channel  may  be  different. With this setup, we present a policy that achieves expected regret of $O (\log{T})$ over a time horizon of duration $T$.  

\end{abstract}

\section{Introduction}

Current spectrum management protocols treat frequency spectrum as a fixed commodity, which leads to spectrum under utilization. Dynamic spectrum access techniques have emerged as good strategies to improve spectrum utilisation. Existing techniques for dynamic spectrum access have focused primarily on the primary/secondary user paradigm, where secondary users detect vacant bandwidths when available and vacate the occupied channel when a primary user wants to transmit.

In this work, we consider the uncoordinated spectrum access model, where there is no such hierarchy among users and users are not allowed to communicate with each other. The users follow a common protocol designed to maximize the system performance rather than an individual's reward. This problem is studied using the stochastic multi-user multi-armed bandit framework. The channels are treated as the arms and the channel gains or the rates received from the channels can be interpreted as the rewards.  

In most of the previous work in this area (eg. \cite{mega}, \cite{mc}, \cite{vvv} and \cite{perchet}), it is assumed that the reward distributions for each channel is the same across all users. In \cite{mega}, \cite{mc} and \cite{perchet}, the assumption is that, in the case of a collision (when more than one user access a channel) all the colliding users receive zero reward, which is the assumption we work with in this paper. The algorithm in \cite{mega} combines an $\epsilon$-greedy approach with a collision avoiding mechanism and achieves expected regret of $O (T^{\frac{2}{3}})$. The musical chairs algorithm in \cite{mc} has an exploration phase where the users estimate the mean rewards for the channels and the number of users, and the users occupy one of the best arms for the remainder of the time horizon. In \cite{vvv}, a model is studied in which more than one user can access a channel simultaneously and receive non-zero rewards and an approach similar to the musical chairs algorithm \cite{mc} is presented. The policies provided in \cite{mc} and \cite{vvv} provide guarantees of constant regret with high probability. The algorithm presented in \cite{perchet} achieves a regret of $O(\log{T})$. The lower bound for the single user case is $O(\log{T})$, and since we are considering a scenario without communication between the users, we expect the regret lower bound to be $O(\log{T})$ at best, which we show is achieved by the policy presented in this paper.

Given that in a practical scenario, users are not colocated in a wireless network, assuming that different users have different reward distributions for the channels results in a more realistic model. There has been some work covering varying reward distributions across users (eg. \cite{dileep}, \cite{hanawal}, \cite{got} , \cite{kaufmann} and \cite{magesh}). The algorithms presented in \cite{dileep} and \cite{hanawal} consider such a model, with an assumption that the players can sense what happens on a channel, \textit{i.e.}, if someone is using the channel or not, or if there is a collision on it. However, such an assumption might be unrealistic for the uncoordinated spectrum access problem. In our work, we consider a fully distributed scenario where players only have access to their previous actions and rewards. The work in \cite{got} considers such a fully distributed setting, and the proposed algorithm achieves a regret of $O(\log^2{T})$. They consider a game-theoretic approach to solve for the pareto-optimal matching that maximizes the system welfare. The work in \cite{magesh} extends this game-theoretic approach to the case of non-zero rewards on collisions. 

The work in \cite{perchet} introduces the idea of using forced collisions as a way to communicate among the users in the homogeneous setting where the reward distributions for the channels are the same across users. In this work, we use the idea of forced collisions in the heterogeneous case. The algorithm presented here was developed independently of the work in \cite{kaufmann}, where an approach similar to ours is explored, \textit{i.e.}, a fully distributed setting with user dependent rewards and with zero reward on collision. However, while the algorithm presented in \cite{kaufmann} achieves logarithmic regret only in the case of a unique optimal matching, the policy presented in this paper results in logarithmic regret even in the case of multiple optimal matchings.  

\section{System Model}

We consider the scenario of a multi-user game involving $K$ users and $M$ channels as the arms in a stochastic multi-armed bandit setup. We assume that $K < M$ and that the rewards for the channels are bounded in $[0,1]$. Let the mean reward for user $j$ on channel $m$ be denoted by $\mu_{j}(m)$. Consider a time horizon $T$, and let the action taken by user (arm chosen by the user) $j$ at time $t \leq T$ be $a_{t,j}$. In the case of a collision, \textit{i.e.}, when multiple users access the same channel, all the colliding users receive zero reward.

Let $\mathcal{A} (K,M)$ denote all the possible user channel matchings, \textit{i.e.}, $\mathbf{a} = [a_1,a_2,...,a_K] \in \mathcal{A} (K,M)$, with $a_j$ denoting the action taken by user $j$. Since we assume that colliding users get zero rewards, we only consider matchings that are unique, \textit{i.e.}, once for which all the users are assigned to distinct arms. Let $\mathbf{a^*} \in \mathcal{A} (K,M)$ be such that $$ \mathbf{a^*} \in \mathop{\arg\max}_{\mathbf{a} \in \mathcal{A}(K,M)} \sum_{j=1}^K \mu_{j}(a_j).$$ Define $J_1$ to be the system reward for the optimal matching, \textit{i.e.}, $J_1 = \sum_{j=1}^K \mu_{j}(a^*_j)$, and $J_2$ to be the system reward for the second optimal matching. In our algorithm, we assume that we have access to a lower bound on the parameter $\Delta$ defined as follows (see also \cite{dileep}, \cite{got}): $$ \Delta = \frac{J_1 - J_2}{2M}$$ Note that this quantity is strictly positive even in the case of multiple optimal matchings. 

\noindent
The expected regret of the system is defined as $$ R(T) = T \sum_{j=1}^K \mu_{j}(a^*_j) - \expect{\sum_{t=1}^T \sum_{j=1}^K \mu_{j}(a_{t,j})} $$
where the expectation is over the actions of the players. 

\section{Algorithm}

We assume that the players are time synchronised and that they enter the system at $t=0$. The algorithm proceeds in epochs for each user. This allows us to proceed without knowing the time horizon $T$. Each epoch has three phases. 

The first phase is the exploration phase, which has two parts. The first part is the fixing phase and is done for each user to obtain a unique ID. Each user accesses the arms uniformly and once a free arm is found, \textit{i.e.}, non-zero reward is received, that arm is played for the rest of the fixing phase. The channel numbers they settle on serve as their IDs. At the end of the fixing phase, the users that are not fixed occupy channel 1, and the fixed users access channel 1 in order of their IDs sequentially. If the fixed users do not face a collision during this step, all users have obtained unique IDs. Once all the users obtain unique IDs, say at epoch $\ell_f$, this part of the exploration phase is no longer done from epoch $\ell_f + 1$. The second part of the exploration phase is for the users to get estimates of the mean rewards of the arms. The users start from the channels corresponding to their IDs, and sample each arm for $\gamma$ time units in a round-robin fashion. 

The second phase is the matching phase and its purpose is for the users to arrive at the optimal matching. The first part of the matching phase is for each user to communicate their estimates of the mean rewards of all the channels to the other users. The users transmit the estimated mean rewards $\hat{\mu}_{j}(m)$ for all the channels in the order of their IDs. Since the players are not allowed to directly communicate with each other, collisions are used as a way to exchange information among players. The use of forced collisions as a form of communication was introduced in \cite{perchet}. The main idea is that there are $M$ channels available to the users and the transmitting user $j$ occupies a certain channel. When the receiving users access the channels one at a time, the channel number on which they face a collision gives them information about what was being transmitted. The value of the estimate $\hat{\mu}_{j}(m)$ that each user has received at the end of the matching phase is denoted by $\doublehat{\mu}_{j}(m)$. Note that this is similar to truncating the value of $\hat{\mu}_{j}(m)$ to a finite number of bits as $\doublehat{\mu}_{j}(m)$. At the end of the first part of the matching phase, $|\hat{\mu}_{j}(m) - \doublehat{\mu}_{j}(m)| \leq \Delta/2$ for $j \in [K]$ and $m \in [M]$ given that all the users have a unique ID. 

Once this communication part of the matching phase is completed, each user has the same approximated values of estimated mean rewards. Each user then independently solves for the set of optimal matchings from the matrix of $\doublehat{\mu}_{j}(m)$ values. If there is a unique optimal matching, the users play the arm according to that matching for the exploitation phase. If there are multiple optimal matchings, it is necessary for each user to choose the same optimal matching from this set. This is achieved in the following manner. The user with ID number one chooses one of the optimal matchings and occupies that channel. The remaining users access the $M$ channels in order of their IDs to get the channel number chosen by the user with ID one and update the set of optimal matchings that correspond with the channel chosen by user one. This process is repeated until all the users settle on the same optimal matching. This matching is played by all the users for the exploitation phase.

In our algorithm, the estimated mean rewards of all the channels are communicated to each user through the idea of forced collisions, and hence all the users can independently solve the assignment problem to arrive at the same set of optimal matchings. However, in \cite{kaufmann}, the communication mechanism is adapted from \cite{perchet}, where a leader-follower protocol is employed. The followers send the values of estimated mean rewards to the leader, and the leader computes the matching that has to be played by the users. While the algorithm presented here gives guarantees of logarithmic regret for the cases of unique optimal matching and multiple optimal matchings, the work in \cite{kaufmann} provides guarantees of logarithmic regret only for unique optimal matchings and quasi-logarithmic regret in the case of multiple optimal matchings. Note that the constants in the upper bound for average regret of our algorithm are comparable to the ones obtained in \cite{kaufmann}.

\begin{algorithm}
\SetAlgoLined
 \textbf{Initialization}: Set $\hat{\mu}_{j}(m) = 0$ for all values of $j \in [K]$ and all $m \in [M]$ and $L_T$ as the last epoch with time horizon T.\\
 \For{$\ell = 1, ..., L_T$}{
  
      \textbf{Exploration phase}: 
      
          Fixing phase : Access channels uniformly for $T_f$ time units till find a channel with no collision; fix on that for remainder of sub-phase. The channel number fixed on serves as unique ID. \\
          If not fixed during fixing phase, send a flag by occupying channel 1.  \\
          If fixed during part fixing phase, access channel 1 in order of ID. Once all users are fixed, skip fixing phase. \\
          Access each channel in a round robin fashion for $\gamma$ time units to get estimates $\hat{\mu}_{j}(m)$.\\
      \textbf{Matching phase}: Enter the matching algorithm to convey the estimates to all users, receive their estimates and calculate the optimal matching.  \\
      \textbf{Exploitation phase}: Occupy the channel resulting from the matching algorithm for $2^\ell$ time units.\\

 }
  \caption{Decentralized MUMAB Algorithm}
\end{algorithm}

\begin{algorithm}
\SetAlgoLined
\textbf{Initialization}: Transmit in order of ID.  

\For{Transmitting user $j' = 1, ..., K$}{

     If turn to transmit, transmit $\hat{\mu}_{j'}(m)$ for $m = 1, ..., M$ as:
     
         Occupy channel $h_1 = \lceil M \hat{\mu}_{j'}(m) \rceil$ for $M$ time units. \\
         Occupy $h_r = \lceil M^r (\hat{\mu}_{j'}(m) - \sum_{n = 1}^{r-1} \frac{h_n - 1}{M^n}) \rceil$ for M time units for each subsequent round $r$ for $\frac{1}{\log {M}} \log{(\frac{1}{\Delta})}$ rounds. \\
     
     If not turn to transmit, in order of IDs:
     
         Access channels in round robin fashion and initialize the estimates of the transmitted value as $ \doublehat{\mu}_{j'}(m) = \frac{2h-1}{2M}$ in the first round. \\
         Update in the subsequent round $r$ as $ \doublehat{\mu}_{j'}(m) = \sum_{n=1}^{r-1}\frac{h_n -1}{M^n} + \frac{2h_r-1}{2M^r}$.

}

Calculate the set $S_M$ of optimal matchings with values of estimates.

\eIf{$S_M$ is a singleton set}{
Assign channel according to the optimal matching.
}{
User 1 chooses one of the optimal matchings.\\
Remaining users update their set of optimal matchings accordingly.\\
Similar mechanism for subsequent users allows users to settle on one of the optimal matchings.
}
  
 \caption{Matching Algorithm for user $i$}
\end{algorithm}

\section{Regret Analysis}

\begin{theorem}
    Assuming the rewards of each channel for all users are bounded in $ [0,1] $ and i.i.d. for all $t \leq T$ for a time horizon $T$ and $\Delta = \frac{J_1 - J_2}{2M}$ is known, the regret of the decentralized MUMAB algorithm is $O(\log{T})$
\end{theorem}

\begin{proof}

The regret incurred during the $L_T$ epochs can be analyzed as the sum of the regrets incurred in the three stages of the algorithm. From the structure of the epochs, we have that $ L_T < \log{T} $. 

\begin{enumerate}
    \item Exploration phase: Let the regret incurred during the exploration phase for all epochs be $R_1$. The exploration goes on for $T_f + K + \gamma M$ time units till epoch $\ell_f$ (when all the users get fixed) and for $\gamma M$ time units after that. Choosing $T_f = M\log (20K)$ and $\gamma = \frac{1}{2\Delta^2}$, we have that 
    \begin{equation}
    \begin{split}
        R_1 &\leq \sum_{\ell = 1}^{L_T}\left( M\left( \frac{1}{2\Delta^2} + 1 +\log {20K}\right)  \right) \\
             &\leq \left( M \left( \frac{1}{2\Delta^2} + 1 +\log{20K} \right) \right) \log{T}.
    \end{split}
    \end{equation}
    
    \item Matching phase: Let the regret incurred during the matching phase be $R_2$. The matching phase runs for $\frac{KM^3}{\log M} \log{\frac{1}{\Delta}} + (M-k)M^2 + M^3$ when there are multiple optimal matchings and for $\frac{KM^3}{\log M} \log{\frac{1}{\Delta}} + (M-k)M^2$ time units when there is an unique optimal matching. Thus  
    \begin{equation}
    \begin{split}
         R_2 &\leq \sum_{\ell = 1}^{L_T} \frac{KM^3}{\log M} \log{\frac{1}{\Delta}} + (M-k)M^2 + M^3 \\
        &\leq \left( \frac{K}{\log M} \log{\frac{1}{\Delta}} + 2 \right) M^3 \log T.
    \end{split}
    \end{equation}
    
    \item Exploitation phase: Let $R_3$ denote the regret incurred during the exploitation phase.  
    Regret is incurred in the exploitation phase in epoch $\ell$ only in case of the following two events: 
    
    \begin{enumerate}
        \item The users do not get a unique ID in the first part of the exploration phase. Let $P_\ell(A)$ denote the probability that after the exploration phase of epoch $\ell$ the users do not have a unique ID  
        \item Given that users have unique IDs, for some $j \in [k] $ and some $m \in [M]$, $|\doublehat{\mu}_{j}(m) - \mu_{j}(m)| > \Delta$. Let $P_\ell(B)$ denote the probability of this event. 
    \end{enumerate}
    
    Thus we have that 
    
    \begin{equation}
        R_3 = \sum_{\ell = 1}^{L_T} 2^\ell(P_\ell(A) + P_\ell(B)).
    \end{equation}
    
    Let $p_f$ denote the probability that with the fixing phase running for $T_f$ time units, all users are fixed. From Lemma 1 of \cite{perchet}, we have that $p_f \geq (1 - K e^{\frac{T_f}{M}})$. From the definition of the event A, we have that 
    
    \begin{equation}
        P_\ell(A) = (1-p_f)^\ell
    \end{equation}
    
    Our choice of $T_f$ results in $p_f > \frac{3(e-1)}{2e}$ and thus, 
    \begin{equation}
        \sum_{\ell = 1}^{L_T} 2^\ell P_\ell(A) = \sum_{\ell = 1}^{L_T} 2^\ell (1-p_f)^\ell \leq \frac{e}{2e-3}.
    \end{equation}
    
    We therefore get the first term of $R_3$ to be bounded by a finite number. 
    
    Let $F$ denote the event that in epoch $\ell$ all the users have unique IDs. This means that in some epoch $\ell_f \in {0,1,...,\ell -1}$, the system was fixed. Thus 
    
    \begin{equation}
    \begin{split}
        &P_\ell(B)\\
        &= P(|\doublehat{\mu}_{j}(m) - \mu_{j}(m)| > \Delta | F) \\
        &= \sum_{i=1}^{\ell -1} P(|\doublehat{\mu}_{j}(m) - \mu_{j}(m)| > \Delta | F, \ell_f =i) P(\ell_f = i) \\
        &\leq \sum_{i=1}^{\ell-1} 2e^{-2 \Delta^2 \gamma (\ell -i)} p_f (1-p_f)^{i-1} \\
        &= 2e^{-2 \Delta^2 \gamma \ell} \sum_{i=1}^{\ell-1} e^{2\Delta^2\gamma i} p_f (1-p_f)^{i-1}.
    \end{split}
    \end{equation}
        
    Choosing $T_f = M\log (20K)$ gives $p_f > \frac{3(e-1)}{2e}$, and using $\gamma = \frac{1}{2 \Delta^2}$ yields
    
    \begin{equation}
        P_\ell(B) \leq \frac{4e}{e-1} e^{-\ell}
    \end{equation}
    
    and 
    
    \begin{equation}
    \begin{split}
        \sum_{\ell=1}^{L_T} 2^\ell P_\ell(B) &\leq \frac{4e}{e-1} \sum_{\ell=1}^{L_T}  \left( \frac{2}{e} \right) ^{-\ell} \\ 
        &\leq \frac{8e}{(e-1)(e-2)}.
    \end{split}
    \end{equation}
     
    Thus 
    
    \begin{equation}
        R_3 \leq C = \frac{e}{2e-3} + \frac{8e}{(e-1)(e-2)}.
    \end{equation}
    
\end{enumerate}

Therefore

\begin{equation}
\begin{split}
    R(T) &= R_1 + R_2 + R_3 \\
    &\leq \left( \frac{M}{2\Delta^2} + \frac{KM^3}{\log M} \log{\frac{1}{\Delta}} + 4M^3\right) \log T + C \\
    &= O(\log{T}).
\end{split}
\end{equation}

\end{proof}

\section{Experimental Results}

In this section, we present experimental results to validate the performance of our algorithm. We applied the proposed algorithm to a system with $K = 10$ users and $M = 10$ channels. Figure 1 shows the plot of average accumulated regret across the time horizon. The algorithm was run for 10 epochs. We see from the figure that the average accumulated regret grows sub-linearly with time and the regret is bounded by $\log{T}$. 

\begin{figure}
    \centering
    \includegraphics[scale = 0.6]{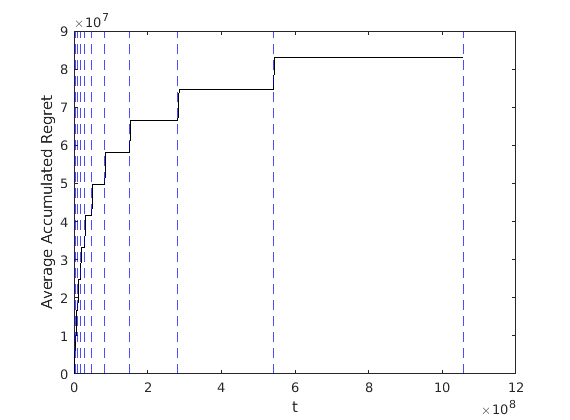}
    \label{fig:my_label}
\end{figure}

\section{Conclusion}

In this work, we have studied the uncoordinated spectrum access problem using a multi-player multi-armed bandit framework where there is no central control and users cannot communicate with each other. We have considered the case where the mean rewards as seen by different users for a particular channel may be different. Under this setup, we have presented an algorithm that provides a regret of  $O (\log{T})$. It is of interest to extend the results presented here to the case where the users receive non-zero rewards on collisions. See \cite{magesh} for an initial study along these lines.

\bibliographystyle{IEEEtran}
\bibliography{refs}

\vspace{12pt}

\end{document}